\documentclass[11pt]{article}
\usepackage{amsmath}   
\usepackage{algorithm}
\usepackage{algpseudocode}
\usepackage{mathrsfs}
\usepackage[utf8]{inputenc}
\usepackage{lmodern}
\usepackage{float}
\usepackage{caption}
\usepackage[normalem]{ulem}
\usepackage{float}
\usepackage[margin=1in]{geometry}
\usepackage{enumitem}
\usepackage{tocloft}
\usepackage{graphicx, titlesec}
\graphicspath{{/figures/}} 
\usepackage{amscd,amsmath,amsthm,amssymb}
\usepackage{verbatim}
\usepackage[all]{xy}
\usepackage{color}
\usepackage{url}
\usepackage[colorlinks=true, linkcolor=purple, citecolor=purple, urlcolor=purple]{hyperref}
\usepackage{booktabs}
\usepackage{tikz-cd}
\usetikzlibrary{patterns}
\usepackage{ytableau}
\usepackage{cleveref}
\usepackage{mathtools}
\usepackage{longtable}
\usepackage{nicematrix}
\usepackage{multirow}
\usepackage{pgf}
\usepackage{todonotes}
\usepackage{xspace,setspace}
\usepackage{bbm,tikz}
\usepackage{algpseudocode}
\usepackage{amsmath,bm}
\usepackage{parallel}

\newcommand{\kb}{\mathbf{k}}

\newcommand{\xb}{\mathbf{x}}

\usepackage{bbm,tikz}
\usepackage{graphicx,subfig}

\def\1{{\mathbbm 1}}

\newcommand{\supp}{{\rm{supp}}}
\def\lcm{{\rm lcm}}

\title{Redundancy analysis using lcm-filtrations: networks, system signature and sensitivity evaluation}

\theoremstyle{plain}%
\newtheorem{theorem}{Theorem}
\numberwithin{theorem}{section}
\newtheorem{remark}[theorem]{Remark}
\newtheorem{definition}[theorem]{Definition}
\theoremstyle{definition}
\newtheorem{example}{Example}
\theoremstyle{remark}
\newtheorem{corollary}[theorem]{Corollary}
\newtheorem{lemma}[theorem]{Lemma}

\author{Fatemeh Mohammadi, Eduardo S\'aenz-de-Cabez\'on, and Henry Wynn}
\date{}

\begin{document}

\maketitle
\begin{abstract}
We introduce the lcm-filtration and stepwise filtration, comparing their performance across various scenarios in terms of computational complexity, efficiency, and redundancy. The lcm-filtration often involves identical steps or ideals, leading to unnecessary computations. To address this, we analyse how stepwise filtration can effectively compute only the non-identical steps, offering a more efficient approach. We compare these filtrations in applications to networks, system signatures, and sensitivity analysis.  
\end{abstract}

\maketitle
\section{Introduction}
The structure of least common multiples of sets of minimal generators of monomial ideals is a key combinatorial tool both in the study of these ideals~\cite{GPV99,MS05,C10,M13} and in their applications. Computing this structure is a computationally demanding task, and it sometimes introduces redundancy that is undesirable in certain situations, particularly when it leads to the expenditure of computing resources on meaningless computations. Some applications of monomial ideals in which their lcm-structure is useful include the analysis of simultaneous failures in coherent systems \cite{MSW17}, multivariate signature analysis \cite{MSW18}, and sensitivity analysis of preference functions \cite{divason2023sensitivity}.

In all of these applications (and their computational aspects), the lcm-filtration was used. The redundancy introduced by this filtration is, in these cases, both a feature and an issue. In this paper, we introduce the stepwise lcm-filtration as a way to handle the least common multiple structure of sets of generators of monomial ideals. This approach can be faster and less redundant but sometimes coarser than the usual lcm-filtration (see Section~\ref{sec:concepts}). Subsequently, we dedicate the rest of the sections to analysing the conditions, in various applications, under which redundancy provides meaningful information and those under which non-redundant computations are more useful or feasible. 

Section~\ref{sec:graphs} focuses on cut ideals of graphs, as used in network theory~\cite{liu2018cut}. We demonstrate that the proposed stepwise filtration exhibits favourable properties in the analysis of sparse networks and graphs, and we describe a phase transition for random graphs concerning the behaviour of the redundancy in the lcm-structure of their corresponding cut ideals. Another application of these ideas pertains to $k$-out-of-$n$ systems and their significant variants: linear and circular consecutive $k$-out-of-$n$ systems. In Section~\ref{sec:simultaneous}, we show that their behaviour with respect to redundancy differs, making the combined use of the stepwise lcm-filtration and the usual one both useful and convenient for the optimal treatment of these systems. 

Finally, Section~\ref{sec:sensitivity} addresses sensitivity analysis, where the stepwise filtration and the usual lcm-filtration complement each other in two ways. First, on a global analysis of the sensitivity of a model, the stepwise filtration avoids processing equal steps in the filtration when such information is unnecessary and also offers a direct interpretation of the filtration steps in a simplicial formulation, enhancing the understanding of factor interactions in simplicial terms. Second, focusing in the local features at each sensitive corner, the redundancy provided by the usual lcm-filtration is better suited to analyse the differences among complexes having the same f-vector and Betti numbers by means of distances between persistence diagrams associated to these filtrations.

\section{Main algebraic concepts}\label{sec:concepts}
\subsection{Stanley-Reisner ideals of simplicial complexes}

A foundational result in combinatorial commutative algebra is the correspondence between the algebraic invariants of monomial ideals and the topological properties of simplicial complexes. This relationship is formalized through Hochster's formula, which connects the multigraded Betti numbers of a monomial ideal to the homology of related simplicial complexes \cite{hochster1977cohen}.

A simplicial complex $\Delta$ on a finite vertex set $V$ is a collection of subsets of $V$, called faces, such that if a set $\sigma$ belongs to $\Delta$, then every subset of $\sigma$ also belongs to $\Delta$. The maximal elements of $\Delta$ with respect to inclusion are called facets, and the dimension of a face is one less than the number of its vertices.

Given a simplicial complex $\Delta$ on $n$ vertices, consider the polynomial ring $\kb[x_1, \dots, x_n]$, where each variable $x_i$ corresponds to vertex $i$. The associated {\it Stanley–Reisner ideal} is defined as  
\[
I_\Delta = \langle \xb^\sigma : \sigma \notin \Delta \rangle,
\]  
where the generators correspond to the minimal non-faces of $\Delta$. Each monomial $\xb^\sigma$ is formed by taking the product of variables indexed by the elements of $\sigma$. By construction, $I_\Delta$ is a squarefree monomial ideal.
Conversely, every squarefree monomial ideal $I$ defines a simplicial complex $\Delta_I$ consisting of all subsets $\sigma$ for which the monomial $\xb^\sigma$ is not in $I$. The two constructions are inverse to each other in the sense that $I_{\Delta_I} = I$.

Hochster’s formula then provides a way to compute the Betti numbers of $I_\Delta$ using the reduced cohomology of subcomplexes induced on subsets $\sigma$.
This connection, which is at the heart of Stanley-Reisner theory, plays a central role in linking monomial ideals and combinatorial topology. However, it applies directly only to squarefree ideals. To apply similar techniques to arbitrary monomial ideals, one uses the polarization process described in \cite{herzog2011monomial, mohammadi2020polarization}. Polarization transforms a monomial ideal into a squarefree ideal in a larger polynomial ring, allowing the application of simplicial methods to the polarized version.

\subsection{The lcm-lattice}

Suppose we have a set of monomials  
$M = \{\xb^{\mu_1}, \dots, \xb^{\mu_r}\}$ in the polynomial ring $\kb[\xb] = \kb[x_1, \dots, x_n]$ over a field $\kb$. For each variable $x_i$, let $\alpha_i$ denote the largest exponent of $x_i$ among the monomials in $M$. The least common multiple (lcm) of the elements in $M$ can then be computed as  
$\lcm(M) = \prod_{i} x_i^{\alpha_i}$.  Recall that the minimal generating set of every monomial ideal is uniquely determined.

\begin{definition}
Let $I$ be a monomial ideal in $\kb[\xb]$ with a minimal monomial generating set $G(I) = \{m_1,\dots,m_r\}$.
Define the set  
$L_I = \{\lcm(\{m_i : i \in \sigma\}) \mid \sigma \subseteq \{1,\ldots,r\}\}$.
The set $L_I$, ordered by divisibility, forms a finite atomic lattice, called the $\lcm$-lattice of $I$.
\end{definition}



The $\lcm$-lattice has been extensively studied in the context of free resolutions of monomial ideals and their connections to atomic lattices. For details, see \cite{GPV99, MS05, P06, M13, C22} and references therein.

\subsection{The \lcm-filtrations of a monomial ideal}
\begin{definition}
Let $I$ be a monomial ideal. A chain of ideals $I_1\subseteq I_2\subseteq\cdots\subseteq I_k=I$ is called a {\em filtration} of $I$. Similarly, a chain of ideals $I=I_1\supseteq I_2\supseteq\cdots\supseteq I_k$ is called a {\em reverse filtration} or a {\em descending} filtration of $I$.
\end{definition}

%

\begin{definition}\label{def:filtration}
Let $I \subseteq \kb[\xb]$ be a monomial ideal, and let $G(I) = \{m_1, \dots, m_r\}$ denote the minimal monomial generating set of $I$.  We define:

\begin{itemize}
    \item[{\rm (i)}]  The $k$-fold $\lcm$-ideal of $I$:   
    For each $k$, let the $k$-fold $\lcm$-ideal of $I$, denoted $I_k$, be the ideal generated by the least common multiples of all sets of $k$ distinct monomial generators of $I$:
    \[
    I_k = \big\langle m_\sigma : \sigma \subseteq \{1, \dots, r\}, \ |\sigma| = k \big\rangle,
    \]  
    where $m_\sigma=\lcm\big(\{m_i\}_{i \in \sigma}\big)$. The sequence of ideals $\{I_k\}$ forms a descending filtration:  
    \[
    I = I_1 \supseteq I_2 \supseteq \cdots \supseteq I_r,
    \]  
    referred to as the \emph{(usual) $\lcm$-filtration} of $I$. 

    \item[{\rm (ii)}]   The stepwise $\lcm$-filtration:   
    This alternative filtration is defined iteratively. Set $\underline{I}_1 = I$, and for each $k \geq 2$, define  
    \[
    \underline{I}_k = \big\langle \lcm(m_i, m_j) : m_i, m_j \in G(\underline{I}_{k-1}), \ i\neq j\big\rangle,  
    \]  
    where $G(\underline{I}_{k-1})$ denotes the minimal generating set of $\underline{I}_{k-1}$.  
\end{itemize}
\end{definition}

\begin{example}\label{rem:finalStep}{\rm 
Observe that the usual and stepwise $\lcm$-filtrations can differ in both the number of ideals and the support of the ideals within them. For example, consider the ideal \( I = \langle abc, bd, cd, e \rangle \subseteq \kb[a, b, c, d, e] \), whose Stanley-Reisner complex \(\Delta_I\) has the set of facets \(\{ab, ac, ad, bc\}\).
The steps of the $\lcm$-filtration are 
\[
I_1=I,\ I_2=\langle abce,bcd,bde,cde\rangle,\ I_3=\langle abcd,bcde\rangle\ \text{and}\ I_4=\langle abcde\rangle.
\]
On the other hand, the steps of the stepwise $\lcm$-filtration are \[
\underline{I}_1=I,\ \underline{I}_2=I_2,\ \text{and}\ \underline{I}_3=\langle bcde\rangle.
\]
}
\end{example}

The $\lcm$-lattice-based filtrations of a monomial ideal $I$ provide valid structural filtrations, because they are based on the $\lcm$-lattice of $I$. The underlying concept of these filtrations is to investigate the changes in the features of the ideal when considering sets of generators instead of individual generators. For example, in \cite{SW09, SW15}, monomial ideals have been used to study failure events and the reliability of coherent systems, where each monomial generator represents a basic working or failure event in the system. To analyze simultaneous events and conduct signature analysis of coherent systems, ideals generated by successively taking least common multiples of a monomial ideal have been applied in \cite{divason2023sensitivity, MSW18, MSW17}.

\subsection{The simplicial form of the stepwise $\lcm$-filtration}
The stepwise $\lcm$-filtration of a monomial ideal $I$ in $\kb[x_1, \dots, x_n]$ is also a reverse filtration, which induces a corresponding simplicial filtration on the Stanley-Reisner simplicial complex $\Delta_I$ associated with $I$. To describe this filtration, let $\underline{\Delta} = \Delta_I = \Delta_{\underline{I}_1}$ be a simplicial complex on the vertex set $\{1,\ldots,n\}$. 
The filtration of $\Delta_I$ is defined as follows:  
Starting with $\underline{\Delta} = \underline{\Delta}_1$, for each $k > 1$, we set  
\begin{equation}\label{eq:simplicial_filtration}
\underline{\Delta}_{k+1} = \underline{\Delta}_k \cup \{F \subseteq \{1,\ldots,n\} : |F^{[|F|-1]} \cap \Delta_k| \geq |F| - 1\}.
\end{equation}
Here, $F^{[|F|-1]}$ denotes the collection of all $(|F|-1)$-subsets of $F$.
In other words, to construct $\Delta_{k+1}$, we add a subset $F \subseteq \{1,\ldots,n\}$ to $\Delta_k$ if and only if all but at most one of its $(|F|-1)$-subsets are in $\Delta_k$.

\medskip
We now demonstrate that the stepwise lcm-filtrations defined on monomial ideals in Definition~\ref{def:filtration}(ii) and on simplicial complexes in \eqref{eq:simplicial_filtration} are compatible.

\begin{theorem}\label{thm:simplicial} Let $\underline{\Delta}$ be a simplicial complex on $\{1,\ldots,n\}$ with the associated Stanley-Reisner ideal $I_{\underline{\Delta}}=\underline{I}_1$. Then for all $k$ we have  $I_{\underline{\Delta}_{k}}=\underline{I}_k$. 
\end{theorem}

\begin{proof}
First note that, as we have
\[
\underline{I}_{k}=\langle {\rm lcm}(m_i,m_j):\ m_i,m_j\in G(\underline{I}_{k-1})\ {\rm and}\ i<j\rangle,
\]
it is enough to prove the statement for $k=2$. Assume that $k=2$. Let $\textsf{x}^F $ be the squarefree monomial ideal with the support $F$. Then:
\begin{align*}
\textsf{x}^F \in I_{\underline{\Delta}_2} 
&\Leftrightarrow F \not\in \underline{\Delta}_2 
\Leftrightarrow \exists\ i, j \in F\ \ \textit{s.t.}\ \ F \backslash i, F \backslash j \not\in \underline{\Delta}_1 \\
&\Leftrightarrow \exists\ F_1, F_2 \subset F\ \ \textit{s.t.}\ \ \textsf{x}^{F_1}, \textsf{x}^{F_2} \in I_{\underline{\Delta}_1} \\
&\Leftrightarrow \exists\ F_1, F_2 \subset F\ \ \textit{s.t.}\ \ \textsf{x}^{F_1}, \textsf{x}^{F_2} \in I_{\underline{\Delta}_1} \ {\rm and}\ {\rm lcm}(\textsf{x}^{F_1}, \textsf{x}^{F_2}) 
\in \underline{I}_2 \\
&\Leftrightarrow \textsf{x}^F \in \underline{I}_2 \ .
\end{align*}
The assertion follows from the above equivalencies.
\end{proof}

Theorem~\ref{thm:simplicial} provides a correspondence between the stepwise lcm-filtration of an ideal $I$ and the filtration of the Stanley-Reisner simplicial complex $\Delta_I$. Consequently, the stepwise lcm-filtration provides a method to construct a filtration of a simplicial complex. To the best of our knowledge, no analogous construction exists for the usual $\lcm$-filtration.

 \subsection{\lcm-filtration versus stepwise lcm-filtration}\label{sub:versus}
The analysis of the lcm structure of a monomial ideal is useful in a wide range of applications. However, the combinatorial nature of these objects makes their size intractable as the number of generators grows. Therefore, different strategies and tools must be used, along with criteria for determining when and why to use each.

\medskip
In general, the lcm-lattice encodes all the interactions between the generators of the ideal. The maximal size of the lcm-lattice for a monomial ideal with $r$ minimal generators is $2^r$, which occurs when the Taylor resolution \cite{T66} of the ideal is minimal, see \cite{A17}. We refer to this as the {\em Taylor lattice} for $r$ generators. Both the usual and stepwise lcm-filtrations consist of a set of ideals whose generators are elements of the lcm-lattice of the ideal. The usual lcm-filtration captures all $k$-fold lcm-ideals and is typically a larger subset of the lcm-lattice than the stepwise lcm-filtration. The relationship between them depends on the specific problem (see \cite{MSW18} and Section~\ref{sec:simultaneous} below). The expected number of generators of the ideal at the $k$-th step in these filtrations is $\textstyle{\binom{r}{k}}$ for the usual lcm-filtration and $\textstyle{\binom{|G(I_{k-1})|}{2}}$ for the stepwise one. Thus, if the information provided by the stepwise filtration is sufficient for the problem at hand, it is crucial to decide which filtration to use and when. In the following sections, we will address this issue in the context of several applications: the cut ideals of networks (Section~\ref{sec:graphs}), signature analysis (Section~\ref{sec:simultaneous}), and consecutive $k$-out-of-$n$ systems (Section~\ref{sec:sensitivity}).

\medskip
Note that, although the stepwise filtration is generally smaller than the usual lcm-filtration, this does not necessarily imply that the former is a subset of the latter, as the following example shows.
\begin{example}\label{ex:differentFiltrations}
    Consider the simplicial complex $\Delta$ with the facets 
    \[\{1, 2, 3\}, \{1, 4\}, \{1, 5\}, \{3, 4\}, \{4, 5\},  \{6, 7\}.\]
The associated ideal is $I=\langle x_2 x_4, x_1 x_3 x_4, x_2 x_5, x_3 x_5, x_1 x_4 x_5, x_1 x_6, \\ x_2 x_6,x_3 x_6, x_4 x_6, x_5 x_6, x_1 x_7, x_2 x_7, x_3 x_7, x_4 x_7, x_5 x_7\rangle.
$
The lcm-filtration consists of 15 steps, while the stepwise lcm-filtration contains only 6 ideals. In the lcm-filtration, we observe the following equalities:  
\[
\underline{I}_1 = I_1, \ \ \underline{I}_2 = I_2, \ \ \underline{I}_5 = I_9 = I_{10}, \ \ \text{and} \ \ \underline{I}_6 = I_{12} = I_{13} = I_{14} = I_{15}.
\]  
However, the ideals $I_3, I_4, I_5, I_6, I_7, I_8$, and $I_{11}$ are distinct from all other $k$-fold ideals and any ideal $\underline{I}_k$. This shows that, although the filtrations of the  complex $\Delta$ using the lcm and stepwise lcm methods eventually coincide, their intermediate steps can differ significantly.
\end{example}



\section{The 
cut ideals of graphs}\label{sec:graphs}
Let $G$ be a graph with the vertex set $V=\{1,\ldots,i\}$ and the edge set $E$.  For a $j$-partition $C=A_1\vert A_2\vert\ldots\vert A_j$ of $V$, consisting of $j$ disjoint nonempty subsets $A_k\subset V$ we define  
\[
E(C) = \{\{a,b\} \in E: a\in A_k,\ b\in A_\ell\text{ for some } 1\leq k<\ell\leq j\}.
\]  

A $j$-partition $C = A_1\vert A_2\vert\ldots\vert A_j$ of the set $\{1, \dots, i\}$ is called a $j$-cut of a graph $G$ if the subgraph induced by each part $A_k$ ($1 \leq k \leq j$) is connected. Let $\mathcal{P}_{i,j}(G)$ denote the set of $j$-partitions of $G$. 
\medskip

Let $\kb$ be a field, and let $R = \kb[x_e : e \in E]$ denote the polynomial ring in the variables corresponding to the edges of $G$. We associate the monomial  
\[
m_C = \prod_{e \in E(C)} x_e
\]  
to each cut $C = A_1\vert A_2\vert\ldots\vert A_j$ of $G$. 
The {\em cut ideal}  $\mathcal{C}_G$ of $G$ is the ideal in $R$ generated by all $m_C$ for $2$-cuts $C$ of $G$.
Let $P_{i,j}$ denote the ideal minimally generated by the monomials $m_C$, where $C$ is a $j$-partition in $\mathcal{P}_{i,j}(G)$. 
In words, the ideal $P_{i,j}$ has a minimal generator corresponding to each $j$-partition $C$ in $\mathcal{P}_{i,j}(G)$. 



\begin{definition} Let $\delta=\delta_1\vert\ldots\vert\delta_k$ and $\tau=\tau_1\vert\ldots\vert\tau_\ell$ be two partitions of the set $\{1,\ldots,i\}$. Then 
\begin{itemize}
\item
$\delta$ and $\tau$ are  {\it compatible} if $k=\ell$ and for each $r$ there exists $s$ such that $\delta_r\subseteq \tau_s$ or $\tau_s\subseteq\delta_r$.

\item
$\delta$ is a {\it refinement} of $\tau$ if $\ell<k$ and for each $r$ there exists $s$ such that $\delta_r\subseteq \tau_s$.

\item the union of two compatible partitions is defined as $\sigma=\sigma_1\vert\ldots\vert\sigma_{d}$ when for each $r$, with $1\leq r\leq d$, there exist $i$ and $j$ such that $\sigma_r=\delta_i=\tau_j$ or $\sigma_r=\delta_i\subset \tau_j$ or $\sigma_r=\tau_j\subset\delta_i$.
\end{itemize}
\end{definition}

We now focus on two extreme families of graphs: complete graphs and trees. In Section~\ref{sec:regime}, we will study the transition between these two extreme cases on random graphs.


\medskip
First, we describe the simplicial complexes in the lcm-filtration of the cut ideals of complete graphs as follows. The result follows directly from the definition, and therefore, we omit the proof.
\begin{corollary}
Let $K_i$ be the complete graph on $i$ vertices. Let $P_{i,j}$ be the ideal generated by the $j$-cuts of $K_i$. Let $\Delta_{i,j}$ denote the associated simplicial complex to $P_{i,j}$. Then the facets of $\Delta_{i,j}$ are corresponding to the complement of the $(i-j+1)$-subsets of $K_i$ with no cycle, i.e. $F$ is a facet of $\Delta_{i,j}$ if and only if 
\[
F=E(K_i)\backslash S,\ \textit{where $|S|=i-j+1$ and $S$ does not contain any cycle.}
\]
In particular, the facets of $\Delta_{i,2}$ are corresponding to the complement of the spanning trees of $K_i$, i.e. $F$ is a facet of $\Delta_{i,2}$ if and only if 
\[
F=E(K_i)\backslash T,\ \textit{where $T$ is a spanning tree of\ } K_i.
\]
\end{corollary}

Consider the complete graph $K_i$ with the vertex set $[i]$. 
For each $j$ with $1 \leq j \leq i$, we denote the $j$-fold $\operatorname{lcm}$-ideal of $I$ with $I_{i,j}$. We have the following relations between the ideals $I_{i,j}$ and the ideals $P_{i,j}$ which gives a nice combinatorial description of the lcm-ideals of the cut ideal of the complete graph $K_i$.

\begin{lemma}\label{lem:two}
Let $\delta=\delta_1\vert\ldots\vert\delta_k$ and $\tau=\tau_1\vert\ldots\vert\tau_k$ be two distinct partitions of $K_i$. Then there exists a $(k+1)$-partition $\sigma$ such that $m_\sigma|m_\delta m_\tau$.
\end{lemma}

\begin{proof}
First note that there exists $r$ and $s$ such that the sets $\delta_r\backslash\tau_s$ and $\delta_r\cap \tau_s$ are  not empty. We set $\sigma=\sigma_1\vert\cdots\vert\sigma_{k+1}$ to be the $(k+1)$-partition with 
\[
\sigma_1=\delta_1,\ldots,\sigma_{r-1}=\delta_{r-1}, \sigma_r=\delta_r\cap\tau_s,\quad\text{and}
\] 
\[
\sigma_{r+1}=\delta_r\backslash\tau_s,\sigma_{r+2}=\delta_{r+1},\ldots,\sigma_{k+1}=\delta_k.
\]
Since $m_{\sigma}=m_{\delta}\prod_{a\in \delta_r\backslash\tau_s\atop  b\in \delta_r\cap\tau_s} x_{ab}$ and  $m_\tau|\prod_{a\in \delta_r\backslash\tau_s\atop  b\in \delta_r\cap\tau_s} x_{ab}$ we have that $m_\sigma|m_\delta m_\tau$.
\end{proof}

Before stating the main result, we recall that the Stirling number of the second kind, denoted by $S(n, k)$, represents the number of ways to partition a set of $n$ elements into $k$ non-empty, disjoint subsets.
\begin{theorem}
Let $P_{i,k+1}$ be the ideal generated by the $(k+1)$-cuts of $K_i$. For all $k\geq 1$, we have $I_{i,2^{k-1}}=\cdots=I_{i,2^k-1}=P_{i,k+1}$. Moreover, the number of generators of each ideal $P_{i,k+1}$ is  
the Stirling number of the second kind, 
$S(n,k
+1)$. 
\end{theorem}

\begin{proof}
We divide the proof into two parts: (1) $P_{i,k+1}\subseteq I_{2^k-1}$ and (2) $I_{i,2^{k-1}}\subseteq P_{i,k+1}$. The assertion then follows, as we have the chain of inclusions $I_{i,2^{k}-1}\subseteq\cdots\subseteq I_{i,2^{k-1}}$. 

\medskip

To prove (1) assume that $\delta=\delta_1\vert \delta_2\vert\dots\vert \delta_{k+1}$ is a $(k+1)$-partition of $K_i$. Then for each subset $A\subset [k]$ we define the $2$-partition
\[
\tau_A=\tau_{A_1}\vert\tau_{A_2}\quad\text{where}\quad \tau_{A_1}=\cup_{a\in A} \delta_a\ \ \text{and}\ \ \tau_{A_2}=\cup_{a\not\in A} \delta_a.
\]
Now it is easy to see that $m_\delta=\lcm(\{m_A\})$, when $A$ runs over all nonempty subsets of $[k]$, which implies (1).

\medskip

To prove (2) we need to show that for any collection  $\mathcal{C}=\{C_1,\ldots,\\ C_{2^{k-1}}\}$ of $2$-partitions there exists a $(k+1)$-partition whose associated monomial divides the $\lcm$-monomial of $\mathcal{C}$. 
The proof is by induction on $k$.  For $k=1$ the proof is clear. Let $k>1$. Then by induction hypothesis there are two $k$-partitions $\delta$ and $\tau$ such that 
\[
m_{\delta}|\lcm(m_{C_1},\ldots,m_{C_{2^{k-2}}})\ \ \text{and}\ \ m_{\tau}|\lcm(m_{C_{2^{k-2}+1}},\ldots,m_{C_{2^{k-1}}})
\]
and so $m_{\delta}m_\tau|\lcm(m_{C_1},\ldots,m_{C_{2^{k-1}}})$. Now by Lemma~\ref{lem:two} we have a $(k+1)$-partition $\sigma$ such that $m_\sigma|m_\delta m_\tau$, as desired.
\end{proof}
The following tables records the number of generators of the ideals $P_{i,j}$, the ideal of the $j$-cuts of $K_i$. for $i$ and $j$ up to $10$. 
\[
\begin{array}{c|ccccccccc}
i \backslash j & 2 & 3 & 4 & 5 & 6 & 7 & 8 & 9 & 10 \\ \hline
2 & 1 \\
3 &  3 & 1 \\
4 &  7 & 6 & 1 \\
5 &  15 & 25 & 10 & 1 \\
6 &  31 & 90 & 65 & 15 & 1 \\
7 &  63 & 301 & 350 & 140 & 21 & 1 \\
8 &  127 & 966 & 1701 & 1050 & 266 & 28 & 1 \\
9 &  255 & 3025 & 7770 & 6951 & 2646 & 462 & 36 & 1 \\
10 & 511 & 9330 & 34105 & 42525 & 22827 & 5880 & 750 & 45 & 1 \\
\end{array}
\]
As an immediate consequence of the above theorem we have:
\begin{corollary}\label{cor:identical}
Let $I$ be the cut ideal of the complete graph $K_i$. Then, the set of ideals $\{I_k\}$ and 
$\{\underline{I}_{k}\}$, appearing in the lcm-filtration and stepwise filtration of $I$ are the same.
\end{corollary}
\begin{remark}
    Corollary~\ref{cor:identical} does not hold in general; see Example~\ref{rem:finalStep}.
\end{remark}


Let us now consider the graph $G = T_i$, a tree on the vertex set $\{1,\ldots,i\}$. The number of edges in $T_i$ is $i-1$, and each edge is a $2$-cut for $T_i$. The ideal $P_{i,2}$ is generated by all the variables, so it has a Taylor lattice, indicating that there is no redundancy in the lattice structure of these ideals, unlike in the case of complete graphs. Furthermore, since $T_i$ has no cycles, any set of $j$ edges of $T_i$ forms a $j$-cut. Hence, we have the following corollary: 

\begin{corollary}\label{cor:identical2}
Let $I$ be the cut ideal of a tree $T_i$. For each $j \leq i$, the ideal $P_{i,j}$ has 
$\binom{i-1}{j}$ generators, each corresponding to a distinct set of $j$ variables. 
In this case, the elements of the lcm-lattice and the ideals in both the usual and stepwise filtrations coincide.
\end{corollary}

We conclude by noting that cut ideals naturally arise in the context of divisor theory (see, for example, \cite{postnikov2004trees, mohammadi2016divisors, mohammadi2014divisors, mohammadidivisors}), and that their Betti numbers and resolutions have been extensively studied.



\subsection{Density regime in cut ideals of random graphs}\label{sec:regime}
Measuring the redundancy of the lcm-filtration with respect to the stepwise version is valuable for certain aspects of analyzing the lcm structure of an ideal or simplicial complex. One such aspect, for example, is computational time. We can distinguish between the two scenarios discussed in the previous section: complete graphs (which exhibit high redundancy) and trees (which exhibit low redundancy).

\smallskip
When the number of generators of an ideal $I_k$ in the lcm-filtration of an ideal $I$ is close to $\textstyle{\binom{r}{k}}$ for every $k$, where $r$ is the number of minimal generators of $I$, the usual lcm-filtration has low redundancy. Hence, its computation is less demanding than that of the stepwise filtration. At step $k$ of the usual lcm-filtration, the procedure computes $\textstyle{\binom{r}{k}}$ least common multiples and then eliminates those generators that are divisible by others. This reduction is small because the actual number of minimal generators is close to $\textstyle{\binom{r}{k}}$. In contrast, the stepwise filtration at step $k$ computes the two-fold lcm of the generators of $I_{k-1}$, which has approximately $\textstyle{\binom{r}{k-1}}$ generators. This requires computing $\textstyle{\binom{\binom{r}{k-1}}{2}}$ least common multiples, a significantly more demanding computation. This case is exemplified by ideals whose Taylor resolution is minimal or close to minimal \cite{A17}. In the squarefree case, this corresponds to little overlap among the supports of the generators.

\smallskip

The opposite situation occurs in the case of high redundancy, i.e., when $|G(I_k)| \ll \binom{r}{k}$, where $G(I_k)$ denotes the minimal generating set of $I_k$. In this case, the usual lcm-filtration will still compute $\binom{r}{k}$ least common multiples, but the number of minimal generators will be significantly reduced. On the other hand, the stepwise lcm-filtration computes the two-fold least common multiples of smaller sets of generators. This typically results in smaller resolutions, which, in the squarefree case, corresponds to dense overlapping among the supports of the generators.

It is thus important, for computational efficiency, to study the transition between these two regimes and identify how to distinguish them before computation. To do so, we will examine cut ideals of random graphs. Consider the complete graph $K_n$ on $n$ vertices. We delete edges from $K_n$ randomly, one at a time. At each step $i$ (i.e., after deleting $i$ edges), we obtain an instance of the Erdős-Rényi random graph model $ER(n,p)$ on $n$ vertices \cite{ER59}, where the parameter $p$ corresponds to $\frac{\binom{n}{2}}{n-i}$, the ratio between the total number of edges in the complete graph and the number of edges remaining in the graph. For each deletion step $i$, we focus on the ratio between the size of the lcm-lattice of the cut ideal $I_{n,i}$ of the resulting graph and the size of the Taylor lcm-lattice for an ideal with the same number of generators $r$ as $I_{n,i}$, i.e., a poset with $2^r$ distinct elements. We consider the density of the graph as the characteristic that relates to this ratio.

\begin{definition}
Let $I$ be a monomial ideal with $r$ generators and let $p(I)$ be the cardinality of its lcm-lattice, $\mathcal{P}(I)$. We define the {\em poset density} of $I$, denoted ${\rm pden}(I)$ as ${\rm pden}(I)=\frac{\vert\mathcal{P}(I)\vert}{2^r}\in\mathbb{R}$.
\end{definition}

\begin{definition}
For a given monotone property $P$ of a family of random graphs, we define $t(n):\mathbb{N}\rightarrow \mathbb{R}$ is a {\em threshold function} if the following conditions hold:
\begin{enumerate}
 \item ${\rm prob}(\mbox{property } P) \simeq 0 \mbox{ if } \frac{p}{t(n)}\rightarrow 0$
\item ${\rm prob}(\mbox{property } P) \simeq 1 \mbox{ if } \frac{p}{t(n)}\rightarrow \infty$,   
\end{enumerate}
where $n$ indicates the number of vertices of the graph and $p$ is the probability function that parametrizes the family of graphs.
If a threshold function $t(n)$ exists for $P$, we say that $P$ undergoes a phase transition at $t(n)$.
\end{definition}

\begin{theorem}\label{th:phaseTransition}
Let $ER(n,p)$ denote the Erd\H{o}s-R\'enyi random graph model, where $n$ is the number of vertices and $p$ is the probability that each possible edge is present in the graph. A threshold function for the poset density of cut ideals of the graphs in $ER(n,p)$ is given by $t(n) = \frac{\lambda}{n}$, where $\lambda$ is a positive constant. That is, the poset density undergoes a phase transition in $ER(n,p)$ at $t(n) = \frac{\lambda}{n}$.
\end{theorem}

\begin{proof}
Let $G$ be a random graph in $ER(n,p)$ for $n \gg 0$ and $I$ be its associated cut ideal. If $G$ has a connected component that contains a constant fraction of the nodes as $n$ grows, we say that $G$ has a giant connected component. In this case, the poset density of $I$ is asymptotically zero, because the giant component is likely to contain many cycles with probability one, leading to high redundancy in the lcm-lattice of $I$. On the other hand, if there is no giant connected component, the graph consists of many small connected components, whose average size approaches one as $n$ grows. In this case, the cut ideal $I$  asymptotically has a Taylor poset.

The presence of a giant component in $ER(n,p)$ corresponds to the threshold function $t(n) = \frac{\lambda}{n}$, as shown in \cite{ER59,N10}, which completes the proof.
\end{proof}

Figure~\ref{fig:ratios-density} illustrates Theorem~\ref{th:phaseTransition} by displaying the results of 10 different random runs of edge deletion for $K_n$, where $n \in \{5,6,7,8\}$. Each line in the graph represents one of these 10 random deletion runs for each value of $n$. We plot the relationship between two ratios: On the vertical axis, we show the ratio of the size of the lcm-poset of the ideals (with size defined as the total number of points in the poset) to the size of the full poset, which is $2^r$ for an ideal with $r$ generators. On the horizontal axis, we plot the density of the graph, defined as the ratio of edges present in the graph to the total number of possible edges in a graph on $n$ vertices, i.e., $\binom{n}{2}$. Each of these graphs corresponds to an instance of the random Erd\H{o}s-R\'enyi model $ER(n,d)$ on $n$ vertices, with probability $d$, where $d$ is the graph's density \cite{ER59}. A future line of research is the study of the density behaviour in families of random monomial ideals \cite{DPSSW19}.

\begin{center}
\begin{figure}[ht]
\subfloat{%
\begin{tabular}{c}
\includegraphics[scale=0.27]{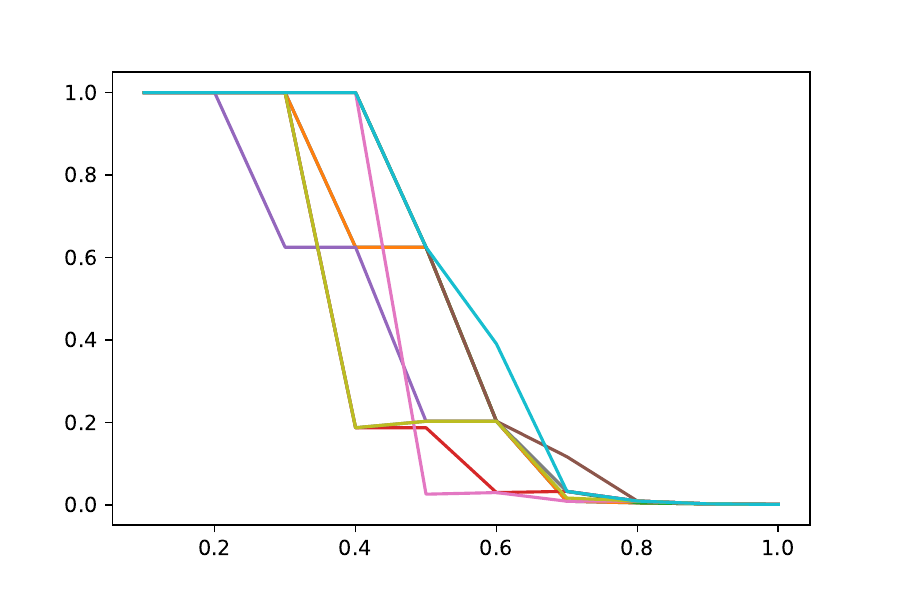} \\
\includegraphics[scale=0.27]{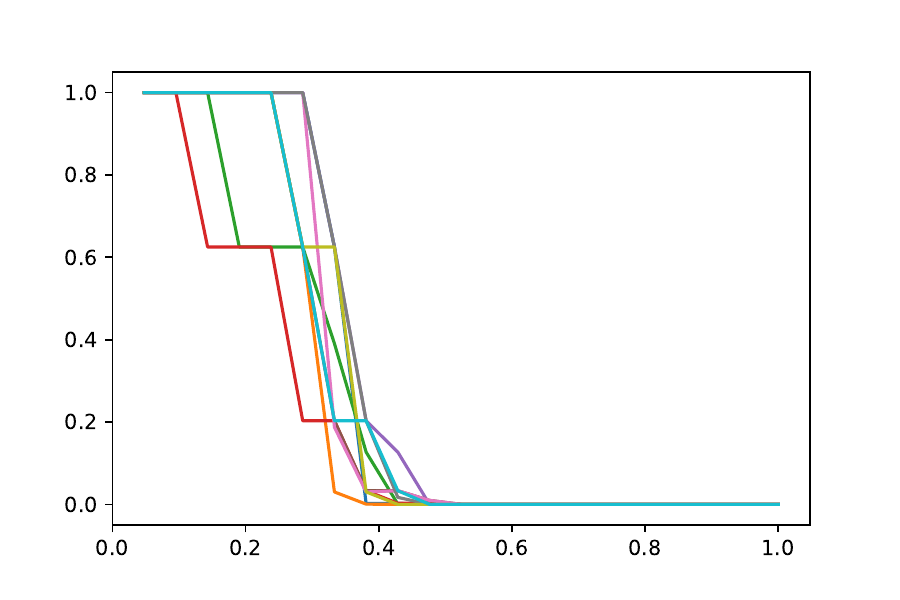}\\
\end{tabular}
}%
\subfloat{%
\begin{tabular}{c}
\includegraphics[scale=0.27]{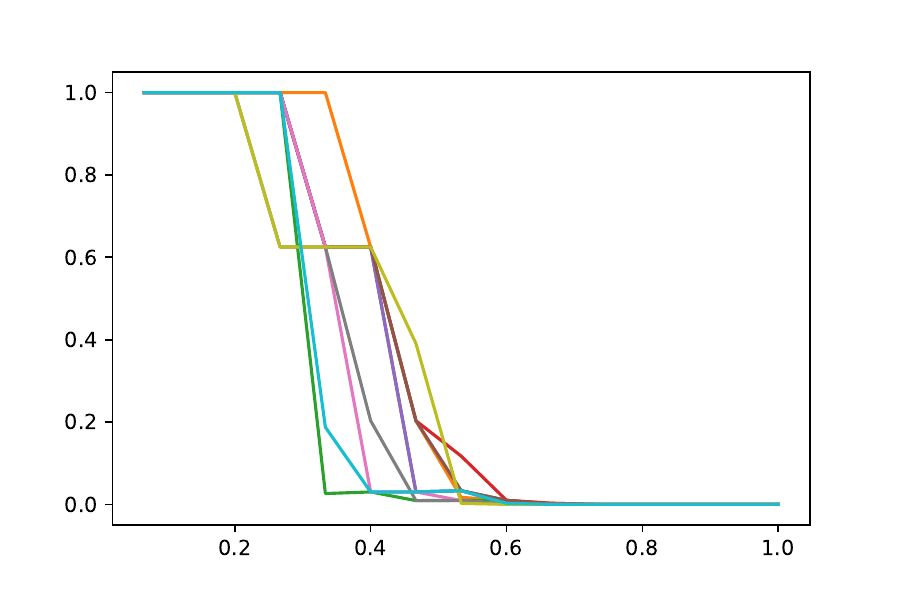} \\
\includegraphics[scale=0.27]{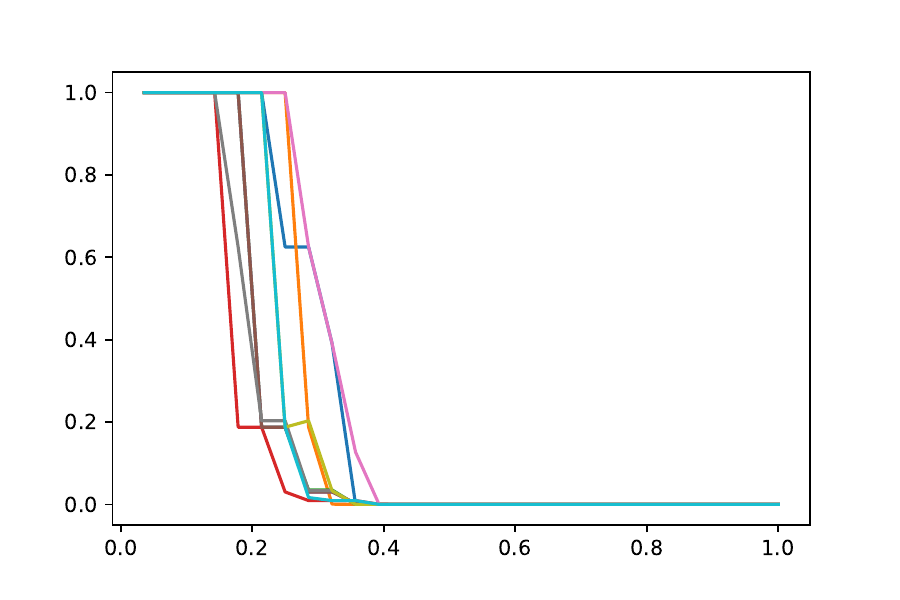}\\
\end{tabular}
}%
\caption{Graph density vs. size of lcm-lattice for cut ideals of subgraphs of the complete graph on $n$ vertices, $n = 5, \dots, 8$.}\label{fig:ratios-density}
\end{figure}
\end{center}


\section{
Simultaneous failures in coherent systems
}\label{sec:simultaneous}
In the fields of systems reliability and signature analysis~\cite{samaniego2007system,boland2001signatures,goharshady2020efficient}, the lcm structure of the associated ideals is fundamental for analyzing simultaneous events. As highlighted in \cite{MSW16,MSW17,MSW18}, a pivotal step in studying simultaneous failures in a system involves computing the various $k$-fold lcm-ideals of its associated failure ideal $I$. 

\smallskip 
Let $n$ be a positive integer, and let $S$ be a coherent system with $n$ components, each taking one of finitely many possible states. The states of the system can be represented as elements of $\mathbb{N}^n$, with a designated subset of failure states $F$, assumed to be coherent (i.e., closed under entrywise ordering). These failure states correspond to the exponents of the monomials in the failure ideal $I$ within $\kb[\xb] := \kb[x_1, \dots, x_n]$.

\smallskip 
Suppose minimal failures can occur simultaneously in the system, and let $Y$ represent the number of such simultaneous failures. The event $\{Y \geq 1\}$ corresponds to system failure, while the intersection of two failure events, $\xb^\alpha$ and $\xb^\beta$, results in $\operatorname{lcm}(\xb^\alpha, \xb^\beta) = \xb^{\alpha \lor \beta}$, representing the event $\{Y \geq 2\}$. This reasoning naturally extends to $\{Y \geq k\}$, enabling the study of tail probabilities $\operatorname{Prob}(Y \geq k)$. Hence, simultaneous failures are encoded within the $k$-fold lcm-ideals $I_k$. 
Finally, consider a probability distribution defined over the system's states. The system's reliability is the probability that it remains in a working state, while its unreliability (or failure probability) is given by $\mathbb{E}[1_F]$, the expected value of the indicator function of the failure set $F$.

In a system with $n$ components that fail independently according to a common failure time distribution and density, the order statistics of the failure times are denoted as $T_{(1)}, T_{(2)}, \dots, T_{(n)}$. The system fails at time $T = T_{(i)}$ for some $i \in \{1, \dots, n\}$, with probabilities:  
\[
s_i = \operatorname{Prob}(T = T_{(i)}) = \operatorname{Prob}(T \geq T_{(i)}) - \operatorname{Prob}(T \geq T_{(i+1)}).
\]
These probabilities, $s_1, \dots, s_n$, can be computed using the failure ideal $I$ and its squarefree representation. Here, $s_i$ represents the conditional probability that exactly $i$ components have failed, given that the system has failed. If $P_i$ denotes the probability of exactly $i$ components failing, then
$s_i = \frac{P_i}{P(F)}$,
where $P(F)$ is the total failure probability. Alternatively, $s_i$ can be expressed as:
\[
s_i = \frac{\mathbb{E}[1_{E_i}(\alpha)]}{\mathbb{E}[1_F(\alpha)]},
\]
where $E_i$ is the set of failure states corresponding to exactly $i$ failed components.

For the case of multiple simultaneous failures, we define $T_k^{(i)}$ as the time of the $k$-th minimal failure, given that $i$ components fail. The $k$-fold signature $s_k^i$ represents the associated probabilities:  
\[
s_k^i = \operatorname{Prob}(T_k = T_k^{(i)}) = \operatorname{Prob}(T_k \geq T_k^{(i)}) - \operatorname{Prob}(T_k \geq T_k^{(i+1)}).
\]
Let $I$ be the failure ideal of the system, and let $I_k$ denote the $k$-fold lcm-ideal of $I$. The $k$-fold lcm-ideal $I_k$ encodes the system states where at least $k$ failures occur simultaneously. We define $I^{[i]}_k$ to be the ideal generated by the squarefree monomials of degree $i$ in $I_k$. 
The $k$-fold signature $s_k^i$ is determined as the difference of the Hilbert series evaluations of the ideals $I_k^{[i]}$ and $I_k^{[i+1]}$:
\[
s_k^i = H(I_k^{[i]}) - H(I_k^{[i+1]}),
\]
where $H(\cdot)$ denotes the Hilbert series. This formulation captures the relationship between the algebraic invariants of the lcm-ideals and the probabilities of simultaneous failures in coherent systems.

In general, the study of simultaneous failure and signature involves the study of $lcm$-filtrations of the corresponding ideals. This might be in general computationally demanding and therefore, under some circumstances, the study of the stepwise filtration is a viable alternative, recall Section \ref{sub:versus}. In the rest of this section we make this observation explicit for one of the most important families of coherent systems, namely $k$-out-of-$n$ systems and variants.

One of the key models in applied reliability engineering is the $k$-out-of-$n$:F system, where F denotes failure \cite{KZ02}. In such a system, failure occurs whenever at least $k$ out of $n$ components fail. A variation, the consecutive linear $k$-out-of-$n$ model, fails whenever $k$ consecutive components fail. The analysis of simultaneous failures in these systems is crucial for many applications \cite{AKO79,MSW18,Y19}. 

The number of minimal failure events in these systems grows exponentially with respect to the number of minimal failures of the system itself. Generally, when $k \ll n$, these systems exhibit low redundancy. However, for the consecutive $k$-out-of-$n$ model, as observed in \cite{MSW18}, this low redundancy makes the problem computationally demanding when using the usual $k$-fold lcm ideal filtration for computation. In some cases, this computation may even exceed the cost of calculating the system's reliability. 

One way to address this challenge is to derive explicit formulas for the generators of these systems, as demonstrated in \cite{MSW17,MSW18}. However, these techniques are not easily adaptable to other types of systems. To overcome this limitation, we analyze the redundant and non-redundant regimes introduced in Section~\ref{sec:regime} and apply this framework also to circular consecutive $k$-out-of-$n$:F systems \cite{PB94,KZ02}, which cannot be efficiently studied using the approaches from \cite{MSW18}.

\bigskip

Figure~\ref{fig:ratios-KN} illustrates the behavior of three types of systems, namely $k$-out-of-$n$, consecutive linear $k$-out-of-$n$, and consecutive circular $k$-out-of-$n$ for $n = 15$ and $k = 2, \dots, 9$. The $k$-out-of-$n$ systems are characterized by $\binom{n}{k}$ generators with a high degree of overlap, resulting in a very low ratio between the size of the full poset on those generators and the size of the lcm-lattice of the ideal. This makes the usual lcm-filtration impractical for large values of $n$. 

In contrast, consecutive $k$-out-of-$n$ systems have fewer generators, leading to smaller total posets. For these systems, the ratio of the full poset size to the actual lcm-lattice size favors the stepwise version of the lcm-filtration when the overlap is moderate. However, for cases of very high or very low overlap, the usual lcm-filtration proves more efficient.

Lastly, the consecutive circular model has a constant number of generators, specifically $n$ generators, for all values of $k$. As the overlap, determined by $k$, increases, the ratio decreases. In this case, the stepwise lcm-filtration is the optimal choice, particularly
since there are no known formulas, 
as in \cite{MSW18}, for the number of generators and Betti numbers of common failure ideals in these systems.

\begin{center}
\begin{figure}[ht]
\includegraphics[scale=0.6]{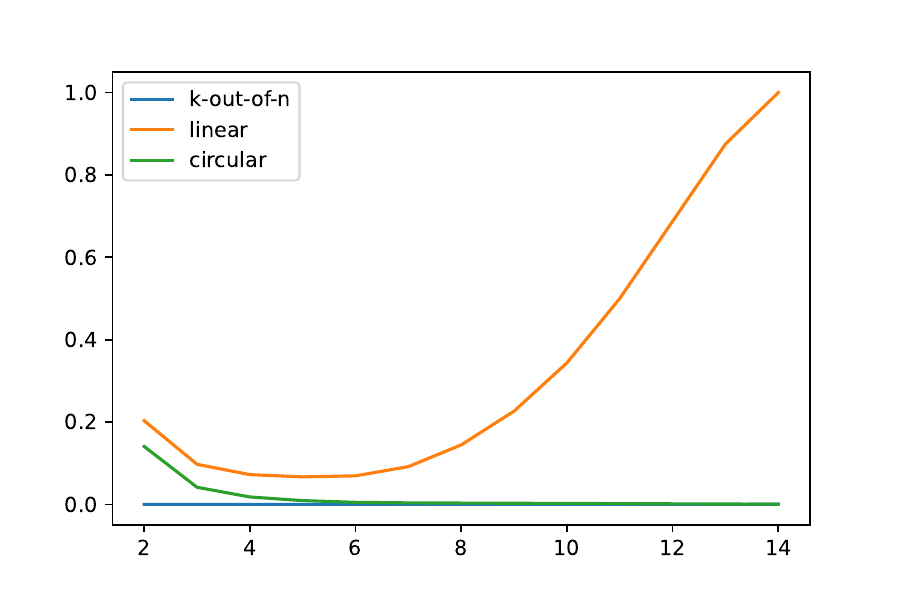} \\
\caption{Size of the lcm-lattice vs. $k$ for $k$-out-of-$15$ ideals, and their linear and circular consecutive variants.}
\label{fig:ratios-KN}
\end{figure}
\end{center}

\section{Sensitivity analysis}\label{sec:sensitivity}


Sensitivity analysis of models studies how different sources of uncertainty in the model inputs influence the uncertainty in the output. It is commonly applied in exploratory modeling, model evaluation, simplification, or refinement. In \cite{divason2023sensitivity}, the authors considered a multi-factor binary decision-making system, and analyzed such systems using an algebraic-combinatorial approach. Monomial ideals and their associated simplicial complexes, were used to identify the input combinations most critical to the decision and warranting careful review. These key combinations were found as the multidegrees at which the ideal representing the system has non-zero Betti numbers, referred to as {\em sensitive corners} of the model. To investigate the variations in the effects of different factor combinations, we analyzed the persistence homology \cite{CZ05} of the simplicial complexes at each sensitive corner, using the standard lcm-filtration.
The method presented in \cite{divason2023sensitivity} can be enhanced by using stepwise lcm-filtration in at least two ways. First, the correspondence between ideals and simplicial complexes in the filtration, outlined in Theorem~\ref{thm:simplicial}, facilitates the interpretation of factor interactions as simplicial complexes. Second, as observed in the signature analysis (Section~\ref{sec:simultaneous}), there is redundancy in the process, where certain simplicial complexes remain unchanged throughout the filtration. While the usual lcm-filtration does not handle this redundancy, the stepwise version efficiently manages it.

Consider a decision system that makes a binary decision based on the satisfaction levels of $n$ factors. The minimum acceptance points of the system are the combinations of factor scores such that any decrease in the satisfaction level of any factor leads to rejection. Let the polynomial ring be $\kb[\xb] := \kb[x_1, \dots, x_n]$. For this system $S$, we associate a monomial ideal $I_S \subseteq \kb[\xb]$, where each variable represents one factor. Each minimal generating monomial of $I_S$ corresponds to a minimal acceptance point of $S$. The multidegrees of $I_S$ at which interactions between minimal acceptance points occur (elements of the lcm-lattice of $I_S$) are referred to as the sensitivity corners of the system or model. At each of these sensitivity corners, with multidegree $\mu$, we can construct a local simplicial complex that models the multi-factor interactions.

\begin{definition}
Let $I$ be a monomial ideal and $\mu \in \mathbb{N}^n$. The upper and lower Koszul simplicial complexes at $x^\mu$ with respect to $I$ are defined as follows:
\[
K^\mu(I) = \{\tau \subset \supp(\mu) : \xb^{\mu - \tau} \in I\}
\]
\[
K_\mu(I) = \{\tau \subset \supp(\mu) : \xb^{\mu' + \tau} \notin I\}
\]
where $\mu'$ is obtained by subtracting one from each nonzero coordinate of $\mu$, i.e., $\mu' = \mu - \supp(\mu)$.
\end{definition}

The relationship between the Betti numbers of $I_S$ at $\mu$ and those of the corresponding Koszul simplicial complexes is given by the well-known Hochster's formula.
\begin{theorem}[cf. \cite{MS05}, Theorem 1.34 and Theorem 5.11] \label{th:HochsterKoszul}
Let $I$ be a monomial ideal and $\mu \in \mathbb{N}^n$. Then the Betti numbers of $I$ are:
\[
\beta_{i,\mu}(I)=\dim_{\kb}\tilde{H}_{i-1}(K^\mu(I))
\]
\[
\beta_{i-1,\mu}(I)=\dim_{\kb}\tilde{H}_{\vert \mu\vert -i-1}(K_\mu(I))
\]
\end{theorem}

We can now apply these tools to perform sensitivity analysis at each corner, either from an ideal-based or simplicial-based perspective. The ideal-based approach is typically more suited for computations, while the simplicial-based approach is better for interpreting the results.

In \cite{divason2023sensitivity}, the primary tool used is persistent homology of the simplicial complexes at the sensitive corners, based on the usual lcm-filtration. The main idea is the following: At each corner of the system ideal, we analyze the Koszul simplicial complex using the lcm or stepwise-lcm filtration of the ideal of the system. Based on this, we cluster these corners to perform a detailed analysis. The analysis in \cite{divason2023sensitivity} (see Section 5 there) shows that this clustering groups the complexes based on size related features, such as number of vertices or number of simplices. This analysis based on the system ideal measures the local effect (at each corner) of interactions between the factors of the minimal acceptance points. In this case, the stepwise filtration offers an efficient alternative to the usual lcm one. However, to distinguish the behaviour of the corners within each cluster, a more detailed local analysis is needed. This can be done by means of the persistence homology of the local simplicial complexes themselves, and in this case, the lcm filtration is advantageous, since it amplifies the distances between persistence diagrams with respect to those obtained . The following example illustrates this idea.

Consider the complexes in Figure~\ref{fig:complexesSameF-B}, which have the same f-vector: $(1,7,7)$ and same Betti numbers: $\beta_0=2$, $\beta_1=2)$. We apply the lcm filtration to them and obtain the persistence diagrams in Figure~\ref{fig:persistenceSameFB}. Applying the stepwise filtration we obtain the diagrams in Figure~\ref{fig:persistenceSWSameFB}, which are much shorter and similar to each other. This can be observed using the distance between persistence diagrams. Tables \ref{table:bottleneck} and \ref{table:wasserstein} show the Bottleneck and Wasserstein distances \cite{CEH07,EH10} between the persistence diagrams of the four complexes in Figure~\ref{fig:complexesSameF-B}. We can see the amplification effect of the lcm-filtration, which is more evident when we use the Wasserstein distance.

\begin{figure}
\begin{tikzpicture}[scale=1]
\filldraw[black] (1,-1) circle (2pt) node(B)[anchor=east]{1};
\filldraw[black] (2,-1) circle (2pt) node(D)[anchor=west]{2};
\filldraw[black] (1,-2) circle (2pt) node(A)[anchor=east]{3};
\filldraw[black] (2,-2) circle (2pt) node(C)[anchor=west]{4};
\filldraw[black] (1,-3) circle (2pt) node(B)[anchor=east]{5};
\filldraw[black] (2,-3) circle (2pt) node(D)[anchor=west]{6};
\filldraw[black] (1,-4) circle (2pt) node(A)[anchor=east]{7};

\draw[black, thick] (1,-1) -- (2,-1);
\draw[black, thick] (1,-1) -- (1,-2);
\draw[black, thick] (2,-1) -- (1,-2);
\draw[black, thick] (1,-2) -- (2,-2);
\draw[black, thick] (1,-3) -- (2,-3);
\draw[black, thick] (1,-3) -- (1,-4);
\draw[black, thick] (2,-3) -- (1,-4);

\filldraw[black] (3,-1) circle (2pt) node(B)[anchor=east]{1};
\filldraw[black] (4,-1) circle (2pt) node(D)[anchor=west]{2};
\filldraw[black] (3,-2) circle (2pt) node(A)[anchor=east]{3};
\filldraw[black] (4,-2) circle (2pt) node(C)[anchor=west]{4};
\filldraw[black] (3,-3) circle (2pt) node(B)[anchor=east]{5};
\filldraw[black] (4,-3) circle (2pt) node(D)[anchor=west]{6};
\filldraw[black] (3,-4) circle (2pt) node(A)[anchor=east]{7};

\draw[black, thick] (3,-1) -- (4,-1);
\draw[black, thick] (3,-1) -- (3,-2);
\draw[black, thick] (4,-1) -- (3,-2);
\draw[black, thick] (3,-2) -- (4,-2);
\draw[black, thick] (4,-1) -- (4,-2);
\draw[black, thick] (3,-3) -- (4,-3);
\draw[black, thick] (4,-3) -- (3,-4);

\filldraw[black] (5,-1) circle (2pt) node(B)[anchor=east]{1};
\filldraw[black] (6,-1) circle (2pt) node(D)[anchor=west]{2};
\filldraw[black] (5,-2) circle (2pt) node(A)[anchor=east]{3};
\filldraw[black] (6,-2) circle (2pt) node(C)[anchor=west]{4};
\filldraw[black] (5,-3) circle (2pt) node(B)[anchor=east]{5};
\filldraw[black] (6,-3) circle (2pt) node(D)[anchor=west]{6};
\filldraw[black] (5,-4) circle (2pt) node(A)[anchor=east]{7};

\draw[black, thick] (5,-1) -- (6,-1);
\draw[black, thick] (5,-1) -- (5,-2);
\draw[black, thick] (5,-2) -- (6,-2);
\draw[black, thick] (6,-1) -- (6,-2);
\draw[black, thick] (5,-2) -- (5,-3);
\draw[black, thick] (6,-2) -- (5,-3);
\draw[black, thick] (6,-3) -- (5,-4);

\filldraw[black] (7,-1) circle (2pt) node(B)[anchor=east]{1};
\filldraw[black] (8,-1) circle (2pt) node(D)[anchor=west]{2};
\filldraw[black] (7,-2) circle (2pt) node(A)[anchor=east]{3};
\filldraw[black] (8,-2) circle (2pt) node(C)[anchor=west]{4};
\filldraw[black] (7,-3) circle (2pt) node(B)[anchor=east]{5};
\filldraw[black] (8,-3) circle (2pt) node(D)[anchor=west]{6};
\filldraw[black] (7,-4) circle (2pt) node(A)[anchor=east]{7};

\draw[black, thick] (7,-1) -- (8,-1);
\draw[black, thick] (8,-1) -- (7,-2);
\draw[black, thick] (7,-2) -- (8,-2);
\draw[black, thick] (8,-1) -- (8,-2);
\draw[black, thick] (7,-2) -- (8,-3);
\draw[black, thick] (8,-2) -- (8,-3);
\draw[black, thick] (7,-3) -- (8,-3);
\end{tikzpicture}
\caption{Four different simplicial complexes on $7$ vertices with same f-vector and Betti numbers.}
\label{fig:complexesSameF-B}
\end{figure}
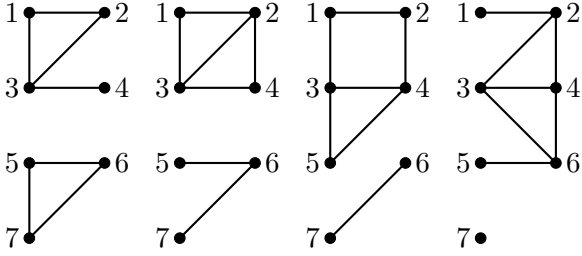

\begin{center}
\begin{figure}[ht]
\subfloat{%
\begin{tabular}{c}
\includegraphics[scale=0.22]{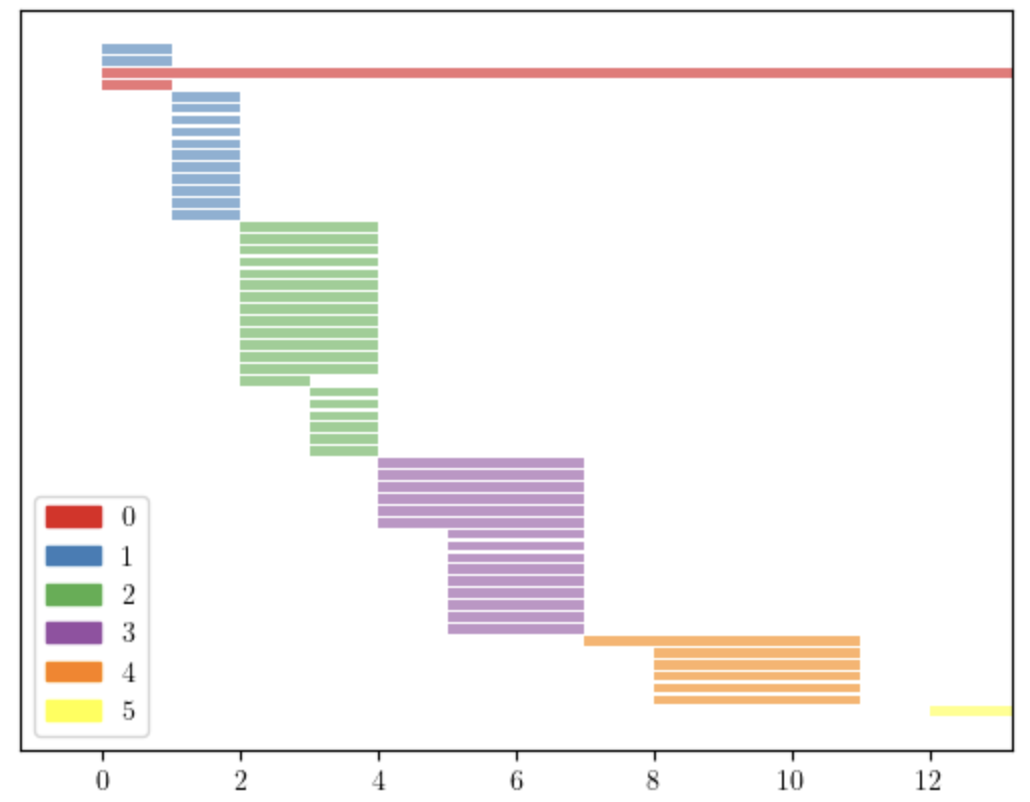} \\
\includegraphics[scale=0.22]{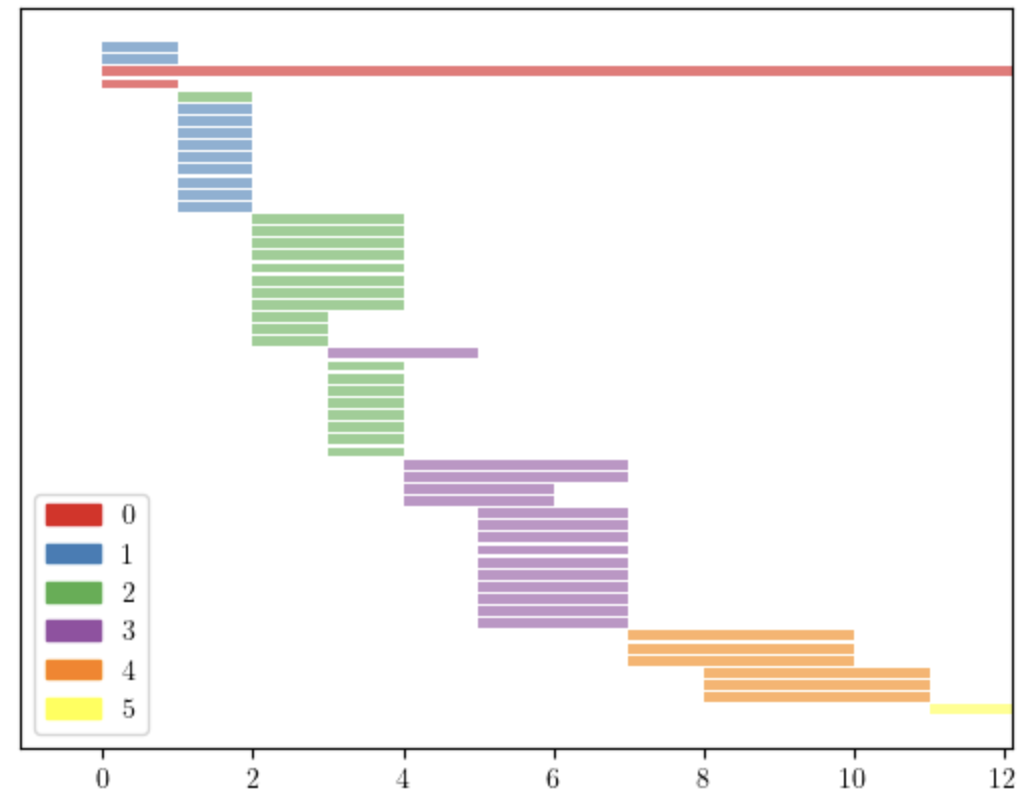}\\
\end{tabular}
}%
\subfloat{%
\begin{tabular}{c}
\includegraphics[scale=0.22]{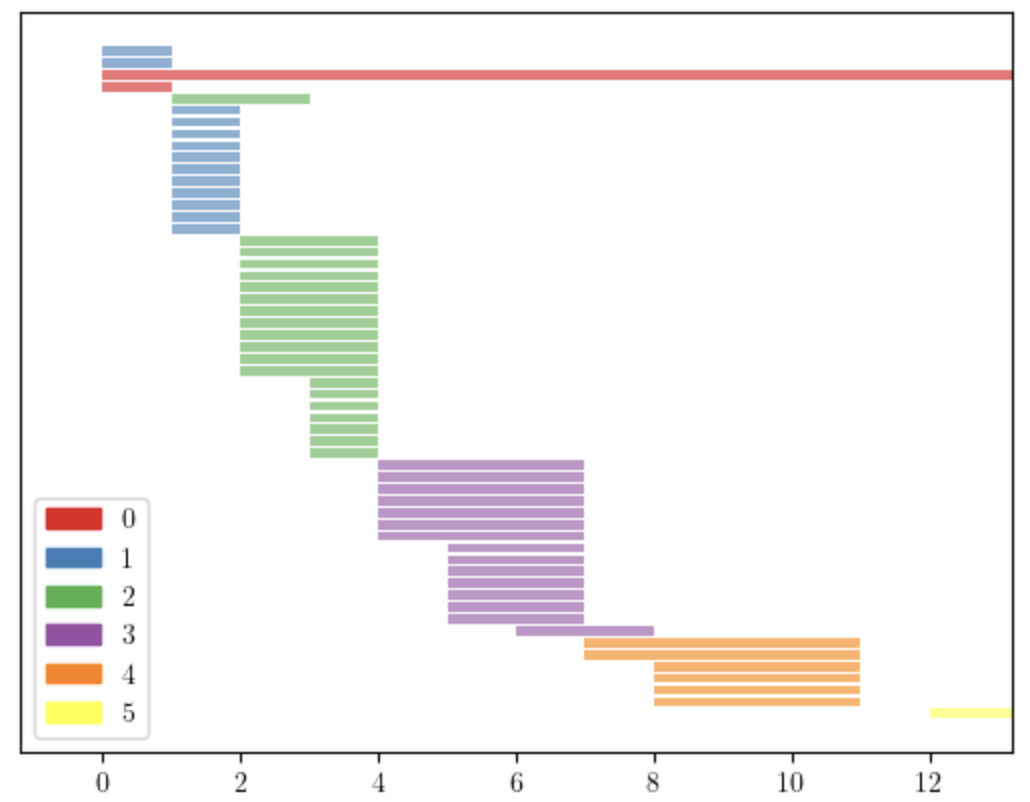} \\
\includegraphics[scale=0.22]{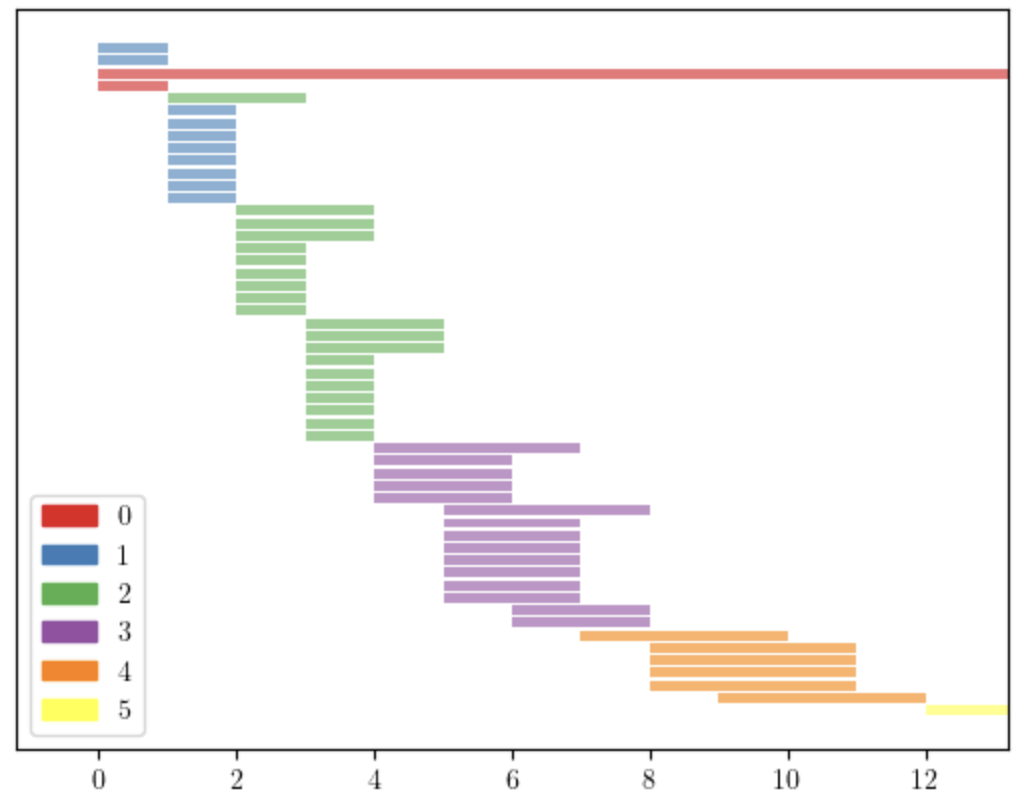}\\
\end{tabular}
}%
\caption{Persistence diagrams of four simplicial complexes with same f-vector and Betti numbers (see Fig. \ref{fig:complexesSameF-B}). Diagrams obtained using the lcm-filtration.}
\label{fig:persistenceSameFB}
\end{figure}
\end{center}

\begin{center}
\begin{figure}[ht]
\subfloat{%
\begin{tabular}{c}
\includegraphics[scale=0.22]{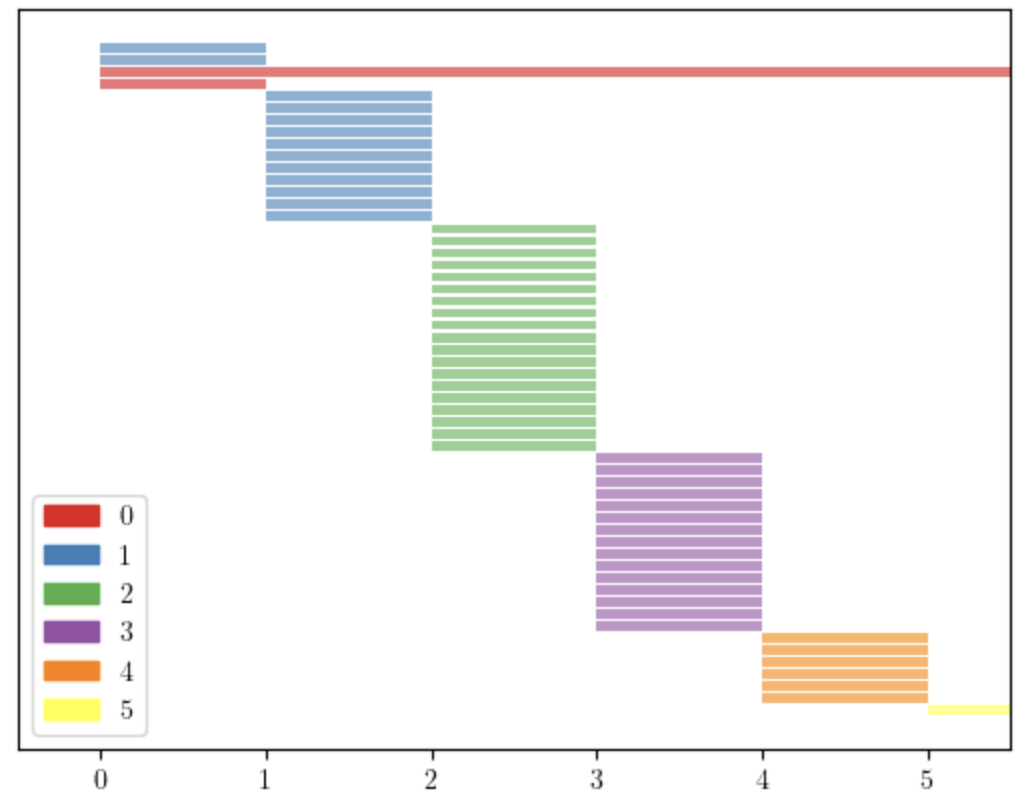} \\
\includegraphics[scale=0.22]{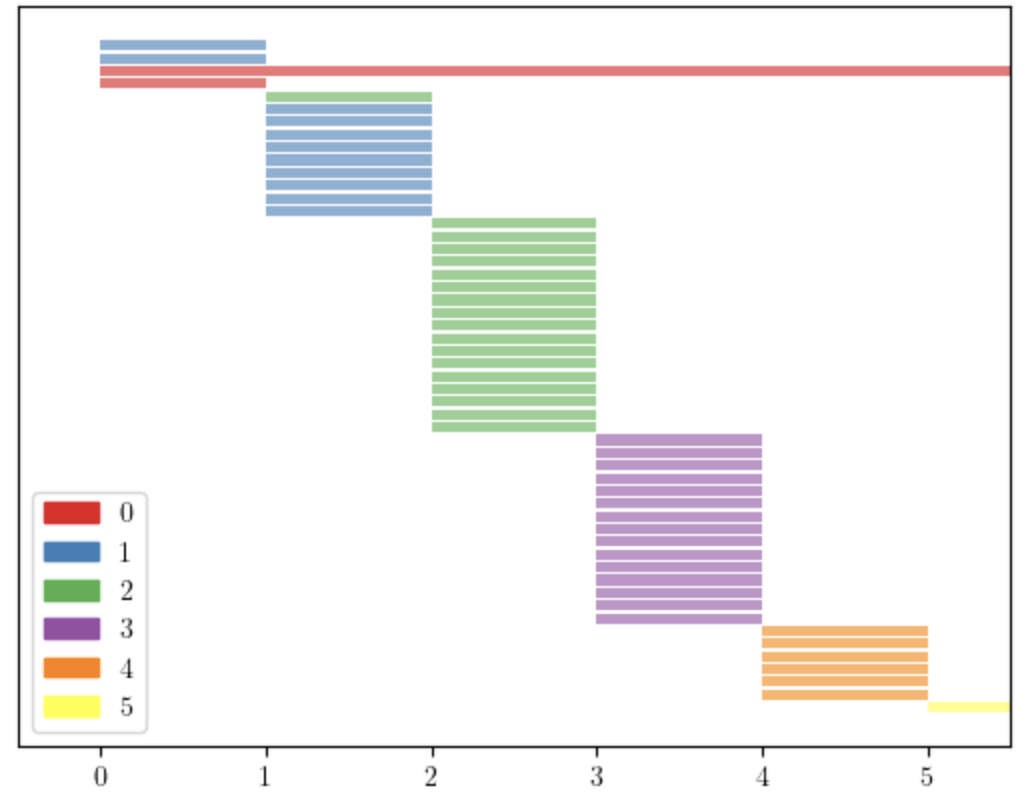}\\
\end{tabular}
}%
\subfloat{%
\begin{tabular}{c}
\includegraphics[scale=0.22]{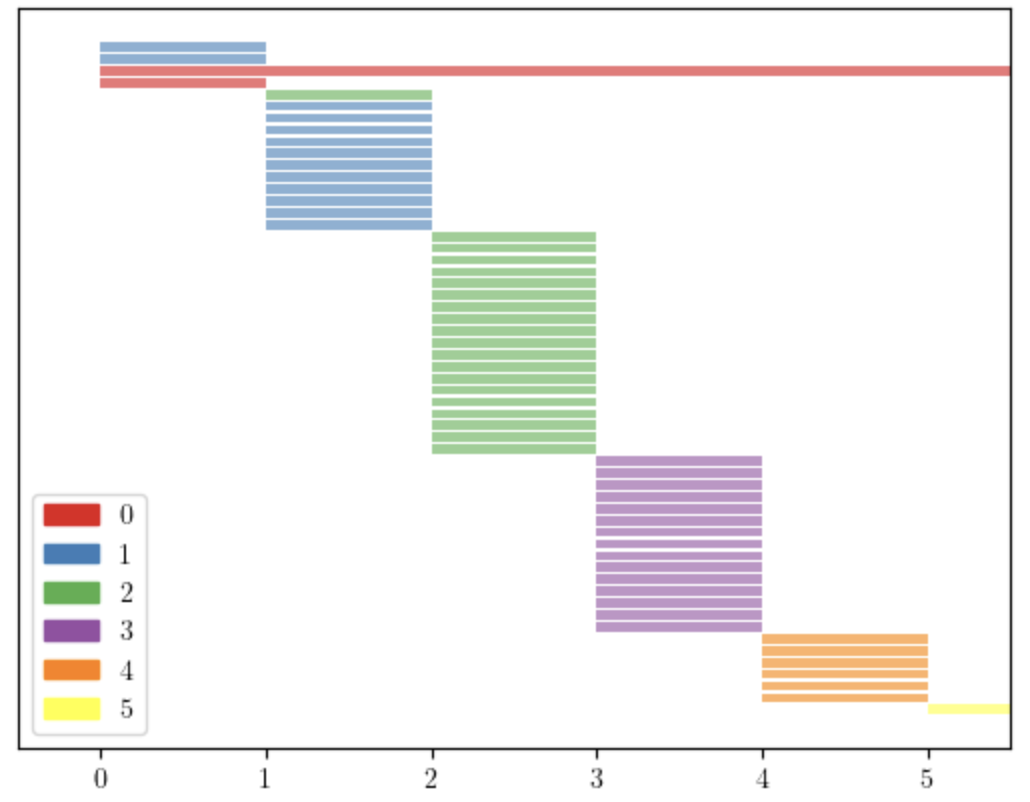} \\
\includegraphics[scale=0.22]{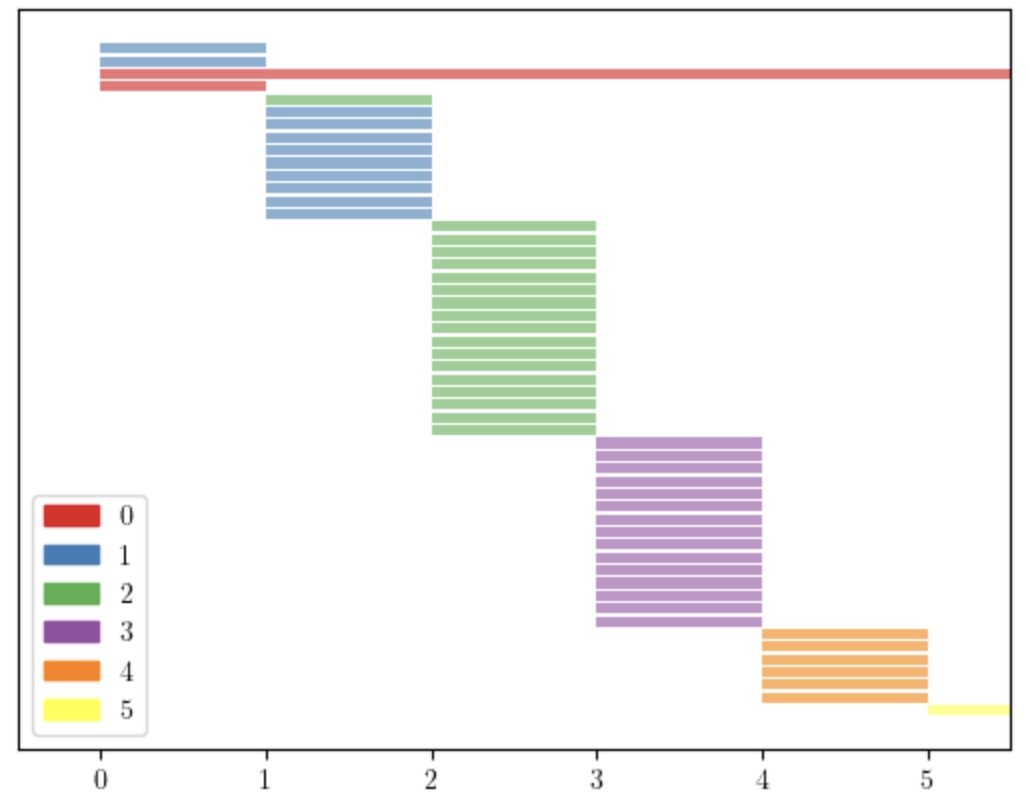}\\
\end{tabular}
}%
\caption{Persistence diagrams of four simplicial complexes with same f-vector and Betti numbers (see Fig. \ref{fig:complexesSameF-B}). Diagrams obtained using the stepwise lcm-filtration.}
\label{fig:persistenceSWSameFB}
\end{figure}
\end{center}

\begin{table}[h]
\begin{tabular}{ cc }   
{\bf Bottleneck Distance} & {\bf Wasserstein Distance} \\  
 \begin{tabular}{c| c c c c} 
 \hline
 &$C_1$ & $C_2$ & $C_3$ & $C_4$ \\ [0.5ex] 
 \hline
 $C_1$ & 0 & 3.0 & 4.5 &3.5\\ 
 $C_2$ & 3.0 & 0 & 4.5 &3.5\\
 $C_3$ & 4.5 & 4.5 & 0 &4.5 \\
 $C_4$ & 3.5 & 3.5 & 4.5 &0\\
 \hline
 \end{tabular} &  
 \begin{tabular}{c| c c c c} 
 \hline
 &$C_1$ & $C_2$ & $C_3$ & $C_4$ \\ [0.5ex] 
 \hline
 $C_1$ & 0 & 5.0 & 15.5 &20.5\\ 
 $C_2$ & 5.0 & 0 & 17.0 &20.5\\
 $C_3$ & 15.5 & 17.0 & 0 &16.5 \\
 $C_4$ & 20.5 & 20.5 & 16.5 &0\\
 \hline
 \end{tabular} \\
\end{tabular}
 \caption{Bottleneck and Wasserstein distances between persistence diagrams of the simplicial complexes in Figure~\ref{fig:complexesSameF-B} using the lcm-filtration}
 \label{table:bottleneck}
\end{table}

\begin{table}[h]
\begin{tabular}{ cc }   
{\bf Bottleneck Distance} & {\bf Wasserstein Distance} \\  
 \begin{tabular}{c| c c c c} 
 \hline
 &$C_1$ & $C_2$ & $C_3$ & $C_4$ \\ [0.5ex] 
 \hline
 $C_1$ & 0 & 0.5 & 1.0 &1.0\\ 
 $C_2$ & 0.5 & 0 & 1.0 &1.0\\
 $C_3$ & 1.0 & 1.0 & 0 &0 \\
 $C_4$ & 1.0 & 1.0& 0 &0\\
 \hline
 \end{tabular} &  
 \begin{tabular}{c| c c c c} 
 \hline
 &$C_1$ & $C_2$ & $C_3$ & $C_4$ \\ [0.5ex] 
 \hline
 $C_1$ & 0 & 0.5 & 2.5 &2.5\\ 
 $C_2$ & 0.5 & 0 & 2.0 &2.0\\
 $C_3$ & 2.5 & 2.0 & 0 &0 \\
 $C_4$ & 2.5 & 2.0 & 0 &0\\
 \hline
 \end{tabular} \\
\end{tabular}
 \caption{Bottleneck and Wasserstein distances between persistence diagrams of the simplicial complexes in Figure~\ref{fig:complexesSameF-B} using the stepwise-filtration}
\label{table:wasserstein}
\end{table}

\section{Conclusions}

In this paper, we introduced two tools for analyzing the least common multiple (lcm) structure of monomial ideals: the lcm-filtration and the stepwise filtration. Both filtrations provide information about the structure of these ideals, as well as the associated simplicial complexes, with broad applicability in various contexts. Specifically, we have explained their use in cut ideals of graphs, simultaneous failures in coherent systems, and sensitivity analysis. In each case, these filtrations allowed us to study the balance between redundancy, efficiency, and information. The tools and applications presented open several avenues for future research on the application of monomial ideals and simplicial complexes in various fields contexts.  

\section*{Acknowledgements}
The research is partially supported by the project INICIA2023/02 funded by La Rioja Government (Spain), as well as the KU Leuven grant iBOF/23/064, and the FWO grants G0F5921N and G023721N.

\bibliographystyle{abbrv}
\bibliography{sample-base}
\end{document}